\documentclass[a4paper,UKenglish,cleveref, autoref, thm-restate]{lipics-v2021}

\usepackage{enumitem}
\usepackage{thmtools}

\usepackage{xspace}

\usepackage{tikz}
\usetikzlibrary{shapes,shapes.geometric,arrows,arrows.meta,fit,calc,positioning,automata}

\theoremstyle{plain}

\theoremstyle{definition}

\theoremstyle{remark}

\theoremstyle{remark}
\newtheorem{notation}[theorem]{Notation}

\theoremstyle{definition}
\newtheorem{problem}[theorem]{Problem}



\makeatletter

\let\c@claim\relax

\makeatother

\makeatletter
\let\c@observation\relax
\makeatother

\declaretheorem[
  style=definition, 
  name=Observation,
  sibling=theorem   
]{observation}

\theoremstyle{plain}
\newtheorem{claim}[theorem]{Claim}

\numberwithin{theorem}{section}

\counterwithin{theorem}{section}
\counterwithin{lemma}{section}
\counterwithin{corollary}{section}
\counterwithin{proposition}{section}
\counterwithin{exercise}{section}
\counterwithin{definition}{section}
\counterwithin{conjecture}{section}
\counterwithin{remark}{section}
\counterwithin{notation}{section}
\counterwithin{note}{section}
\counterwithin{claim}{section}
\counterwithin{figure}{section}
\counterwithin{example}{section}
\counterwithin{observation}{section}

\newcommand{\blue}[1]{\textcolor{blue}{#1}}

\newcommand{\bb}[1]{\mathbb{#1}}
\newcommand\nindent{.5pt}
\newcommand\noverline[1]{%
  \kern\nindent\overline{\kern-\nindent#1\kern-\nindent}\kern\nindent}

\newcommand{\cl}[1]{\mathcal{#1}}

\newcommand{\Roxy}{\text{Roxy}\xspace}
\newcommand{\Sam}{\text{Sam}\xspace}
\newcommand{\DPBG}{\text{DPBG}\xspace}

\newcommand{\CPBG}{\text{CPBG}\xspace}

\newcommand{\txt}[1]{\text{#1}\xspace}

\newcommand{\compclass}[1]{\textsf{#1}\xspace}
\newcommand{\NP}{\compclass{NP}}
\newcommand{\coNP}{\compclass{coNP}}

\newcommand{\PP}{\compclass{P}}

\newcommand{\EXP}{\compclass{EXP}}

\newcommand{\gamesabbr}[1]{\textsl{#1}}

\newcommand{\Win}{\gamesabbr{Win}}
\newcommand{\bidmech}{\gamesabbr{Bid}}
\newcommand{\History}{\gamesabbr{History}}

\newcommand{\chargevec}[2]{\blue{\begin{bmatrix} #1\\
                            #2 
                            \end{bmatrix}}}

\newcommand{\pathbetween}[2]{#1 \rightsquigarrow #2}
\newcommand{\connTo}[2]{#1_{\rightsquigarrow #2}}

\newcommand{\term}[1]{\textcolor{brown}{\textsl{#1}}}
\newcommand{\mathterm}[1]{\textcolor{brown}{{#1}}}

\newcommand{\explayseq}{\Pi}

\newcommand{\ppathseq}{\textbf{\textsc{v}}}
\newcommand{\pbold}[1]{\textbf{\textsc{#1}}}

\newcommand{\ppath}{\textbf{\textsc{v}}}
\newcommand{\pBz}{\textbf{\textsc{B}}_r}
\newcommand{\pBo}{\textbf{\textsc{B}}_s}
\newcommand{\pB}[1]{\textbf{\textsc{B}}_{#1}}
\newcommand{\pseq}[2]{\textbf{\textsc{#1}}_{#2}}
\newcommand{\Rz}{R_r}
\newcommand{\Ro}{R_s}
\newcommand{\Ri}{R_\iota}
\newcommand{\rz}{r_r}
\newcommand{\ro}{r_s}
\newcommand{\ri}{r_\iota}

\newcommand{\Br}{B_r}
\newcommand{\Bs}{B_s}

\newcommand{\flBrack}[1]{\lfloor #1 \rfloor}
\newcommand{\clBrack}[1]{\lceil #1 \rceil}

\newcommand{\boxedpair}[2]{%
  \mathpalette\boxedpairaux{{#1},{#2}}%
}

\newcommand{\boxedpairaux}[2]{%
  \begingroup
  \setlength{\fboxsep}{1pt}%
  \colorbox{gray!20}{$#1(#2)$}%
  \endgroup
}

\newcommand{\twopairarrowbox}[5]{%
  \mathrel{\xrightarrow{%
    \;%
    \ifstrequal{#1}{L}
      {\boxedpair{#3}{#2},\; (#5,#4)}
      {(#3,#2),\; \boxedpair{#5}{#4}}%
    \;%
  }}%
}

\newcommand{\dogexboxedpair}[2]{%
  \mathpalette\dogexboxedpairaux{{#1},{#2}}%
}

\newcommand{\dogexboxedpairaux}[2]{%
  \begingroup
  \setlength{\fboxsep}{1pt}%
  \colorbox{gray!20}{$#1\text{\faDog}(#2)$}%
  \endgroup
}

\newcommand{\catexboxedpair}[2]{%
  \mathpalette\catexboxedpairaux{{#1},{#2}}%
}

\newcommand{\catexboxedpairaux}[2]{%
  \begingroup
  \setlength{\fboxsep}{1pt}%
  \colorbox{gray!20}{$#1\text{\faCat}(#2)$}%
  \endgroup
}

\newcommand{\extwopairarrowbox}[5]{%
  \mathrel{\xrightarrow{%
    \;%
    \ifstrequal{#1}{L}
      {\catexboxedpair{#3}{#2},\; \text{\faDog}(#5,#4)}
      {\text{\faCat}(#3,#2),\; \dogexboxedpair{#5}{#4}}
    \;%
  }}
}

\def\vecdist{0.9cm}
\def\leftvecdist{0.8cm}
\def\labeldist{1.0cm}
\def\gamesdist{5.0cm}
\def\nodedist{2.0cm}
\def\largegamesdist{{1.8*\gamesdist}}
\nolinenumbers
\title{Bidding Games with Rewards: Taming Infinite Configuration Space}

\author{Matan Pinkas}{The Stein Faculty of Computer and Information Science, Ben-Gurion University of the Negev}{Matan.Pinkas@gmail.com}{https://orcid.org/0009-0003-3082-4043}{Supported by ISF grant 2507/21}

\authorrunning{J. Open Access and J.\,R. Public} 

\Copyright{Jane Open Access and Joan R. Public} 





\category{} 

\relatedversion{} 

\EventEditors{John Q. Open and Joan R. Access}
\EventNoEds{2}
\EventLongTitle{42nd Conference on Very Important Topics (CVIT 2016)}
\EventShortTitle{CVIT 2016}
\EventAcronym{CVIT}
\EventYear{2016}
\EventDate{December 24--27, 2016}
\EventLocation{Little Whinging, United Kingdom}
\EventLogo{}
\SeriesVolume{42}
\ArticleNo{23}

\begin{document}

\maketitle

\begin{abstract}
    Bidding games are graph games in which a token is placed on a vertex, each player starts with an initial budget, and a simultaneous auction determines which player moves the token; the players' budgets are then updated accordingly.
    Motivated by scenarios such as resource-allocation systems in which agents receive periodic rewards (e.g., credits, energy) while competing for control, we introduce and study \emph{bidding games with rewards}, in which, at each vertex, players may receive additional budget, incentivizing desired behaviors.
    We focus on \emph{reachability discrete poorman bidding games with rewards} (\txt{DPBGr}).
    The main challenge when compared to discrete bidding games without rewards is that the configuration graph is infinite.
    To this end we introduce a novel technique to eliminate plays with suboptimal infixes. This enables focusing on a finite part of the infinite configuration graph in order to solve the game via approximation to continuous bidding games with overall complexity in \EXP. Finally, we discuss a new type of strategy, usable on a subclass of \txt{DPBGr}, which guarantee a winning strategy for the Reachability player. Membership in this subclass is shown to be in \NP.  
\end{abstract}

\section{Introduction}
\label{sec:intro}

Bidding games are graph games in which a token is placed on a vertex of a directed graph and each player is endowed with an initial budget.
Traversal of the graph is determined by a sequence of bidding rounds, in which the players simultaneously bid for the right to choose the next vertex, and the winner of the bidding selects the successor.
The outcome of a round also determines how the players' budgets are updated, and different \emph{bidding mechanisms} give rise to different classes of games.
In addition, bidding can be either \emph{discrete}, where bids range over $\bb{N}$ (e.g., coins), or \emph{continuous}, where arbitrarily small bids are allowed.
Several mechanisms have been studied in the literature, including \emph{Richman bidding}, where the winner pays its bid to the loser~\cite{DiscreteRichman2010, Richman1999}; \emph{poorman bidding}, where the winner pays its bid to the bank~\cite{Richman1999}; \emph{taxman bidding}, where a fraction of the bid is transferred to the opponent and the remainder is discarded; and \emph{all-pay bidding}, where both players pay their bids regardless of the outcome~\cite{Survey}.
These models provide a natural framework for reasoning about systems in which control is resolved through competition for limited resources, and have been studied for a variety of qualitative objectives, such as reachability and parity.
A central question in this line of work is the computation of \emph{thresholds}, which determine the minimal initial budget required for a player to guarantee winning the game.

Motivated by scenarios such as resource-allocation systems in which agents receive periodic rewards (e.g., credits, energy, or tokens) while competing for control, we introduce and study \term{bidding games with rewards} (\term{\txt{BGr}}), in which, at each vertex, players may receive additional budget, incentivizing desired behaviors.
We focus on \term{reachability discrete poorman bidding games with rewards} (\term{\txt{DPBGr}}), a natural extension of the classical model that combines both consumption and replenishment of resources. The reachability objective captures the fundamental task of steering the system toward a desired set of target states, and serves as a basic building block for more expressive specifications.
Moreover, poorman bidding naturally models settings in which resources are consumed rather than transferred between agents, and thus reflects scenarios where budget represents expendable quantities such as energy or time.
Together, these choices yield a practically motivated setting that isolates the interplay between resource consumption, replenishment, and strategic control.

We begin with a simple example without rewards, which serves to introduce our notation and illustrate the basic dynamics of bidding games.
Henceforth, we refer to Player~0 (the reachability player) as \emph{Roxy} and to Player~1 (the safety player) as \emph{Sam}.
Note that while in continuous bidding games one may assume no ties are made, in discrete bidding games, one must specify how ties are broken (this choice is important as some tie breaking mechanisms result in games that are not determined, see \cite{determinacy}).  
Throughout this paper, we do not explicitly mention what tie breaking mechanisms are used and assume implicitly that a mechanism for which the game is determined is used (see \cref{dpbgr-determined} for examples of \txt{DBPGr} that are determined).


\begin{example}[Simple \txt{DPBG}]\label{ex:simple}
    Consider the \txt{DPBG} $\cl{G}_1$ of \Cref{fig:first} where ties are broken in Sam's favor, and the initial configuration $c_0=(v_0,1,1)$, in which the token is at $v_0$ and both players have budget~$1$.
    \Roxy wins this game as follows.
    She bids $(t,1)$ at $v_0$.
    If \Sam bids less, she wins immediately.
    Otherwise, \Sam bids $1$ and chooses $v_1$; since ties are broken in his favor, he wins the bidding and the next configuration is $c_1=(v_1,1,0)$.
    From this configuration, \Roxy again bids $(t,1)$.
    Since \Sam can only bid $0$, she wins the bidding, the token reaches $t$, and she wins.

    We depict this play as follows, where transitions between configurations are labeled with the players' bids, and the winning bid is shaded gray:
    \begin{center}
    \scalebox{0.9}{
    \[
    (v_0,\overset{\text{\faCat}}{1},\overset{\text{\faDog}}{1})\xrightarrow[\colorbox{gray!20}{\text{\faDog $~v_1,1$}}]{\colorbox{white!20}{\text{\faCat $~t,1$}}}(v_1,1,0)\xrightarrow[\colorbox{white!20}{\text{\faDog $~v_0,0$}}]{\colorbox{gray!20}{\text{\faCat $~t,1$}}}(t,0,0)
    \]
    }
        \end{center}

    Note that \Roxy's (resp.\ \Sam's) budgets and bids are placed above (resp.\ below) the arrow, as annotated here just this once for clarity.

    With regard to the \emph{threshold problem}, let $\mathterm{th_v(B_1)}$ denote the budget \Roxy requires in order to win from $v$ when \Sam\ has budget $B_1$.
    In this game, we have $th_{v_0}(1)=1$, as \Roxy can win with budget~$1$, as shown above, whereas she cannot win with budget~$0$, since she must win at least one bidding round.
\end{example}

Let us now illustrate how adding rewards may change the game.
Our first example shows that adding rewards may change the winner of the game.

\begin{example}[Adding rewards may change outcomes]
Continuing \Cref{ex:simple}, consider the \txt{DPBGi} $\cl{G}'_1$ of \Cref{fig:first}, in which at node $v_1$, \Sam receives a reward of $2$ (and \Roxy receives $0$).
Henceforth, rewards vectors are depicted in blue next to each vertex, with \Roxy's (resp.\ \Sam's) rewards shown on top (resp.\ bottom).
For clarity, we omit vectors in which both players receive zero rewards.
In this game, \Sam can win by repeatedly cycling between $v_0$ and $v_1$, a loop that can be sustained indefinitely due to the rewards he receives.
That is, the following play can repeat forever:
\begin{center}
    \scalebox{0.9}{
\[
    (v_0,1,1)\xrightarrow[\colorbox{gray!20}{\text{$v_1,1$}}]{\colorbox{white!20}{\text{$t,1$}}}(v_1,1,0)\xrightarrow[\colorbox{gray!20}{\text{ $v_0,1$}}]{\colorbox{white!20}{\text{$t,1$}}}(v_0,1,1)
    \]
}\end{center}
\end{example}

Next, we show that adding rewards may change the optimal strategy.

\begin{example}[Adding rewards changes winning strategies]
Consider the \txt{DPBGr} of \Cref{fig:second}, call it $\cl{G}'_2$, and let $\cl{G}_2$ be the corresponding \DPBG obtained by removing the rewards.
We first analyze the game without rewards.
Reaching $t$ via $v_4$ requires \Roxy to win at least one bidding round, as otherwise \Sam can remain in $v_4$ indefinitely without depleting his budget (since ties are broken in his favor).
Similarly, reaching $t$ via $v_1$ requires \Roxy to win at least three bidding rounds.
Accordingly, in $\cl{G}_2$, \Roxy wins from $(v_0,2,0)$ by moving to $v_4$.

In contrast, in $\cl{G}'_2$, this strategy is losing for \Roxy, as she cannot outbid \Sam at $v_4$ in the presence of rewards.
However, she can win by moving to $v_1$, where she receives a reward of $4$, enabling her to win three consecutive bidding rounds and reach $t$.
\end{example}

\begin{figure}
    \centering
    \begin{minipage}[b]{0.48\textwidth}
    \centering
        \scalebox{0.75}{
        \begin{tikzpicture}[->,>=stealth',shorten >=1pt,auto,node distance=\nodedist,semithick,initial text=, initial left, font=\large]
    
        \node[] (v0) {};
        \node[] (v1) [above right of=v0] {};
        \node[]  (v2) [below right of=v0]  {};
        \node[] (t) [below right of=v1] {};
        
        \node[state] (u0) [left of=v1, node distance={\largegamesdist}] {$v_0$};
        \node[state] (u1) [left of=v2, node distance={\largegamesdist}] {$v_1$};
        \node[state,accepting]  (ut) [left of=t, node distance={\largegamesdist}]  {$t$};
        \node[] (g1) [below of=u1, node distance=\labeldist] {$\cl{G}_1$};
        
        \path (u0) edge [bend right] (u1);
        \path (u1) edge  (u0);
        \path (u0) edge (ut);
        \path (u1) edge (ut);
        
        \node[state] (uu0) [left of=v1, node distance=\gamesdist] {$v_0$};
        \node[state] (uu1) [left of=v2, node distance=\gamesdist] {$v_1$};
        \node[state,accepting]  (uut) [left of=t, node distance=\gamesdist]  {$t$};
        \node[] (gg1) [below of=uu1, node distance=\labeldist] {$\cl{G}'_1$};
        \node[] (vec) [left of=uu1, node distance={\leftvecdist}] {$\chargevec{0}{2}$};
        \node[] (stam) [left of=uu1, node distance={1.6*\leftvecdist}] { };
        \node[] (cat) [above of=stam, node distance={0.21cm}] {\textcolor{blue}{\faCat}};
        \node[] (dog) [below of=stam, node distance={0.21cm}] {\textcolor{blue}{\faDog}};

        \path (uu0) edge [bend right] (uu1);
        \path (uu1) edge  (uu0);
        \path (uu0) edge (uut);
        \path (uu1) edge (uut);
        \end{tikzpicture}
        }
    \caption{Adding rewards changes winner}
    \label{fig:first}
    \end{minipage}
    \hfill
    \begin{minipage}[b]{0.48\textwidth}
    \centering
        \scalebox{0.75}{
        \begin{tikzpicture}[->,>=stealth',shorten >=1pt,auto,node distance=2.0cm,semithick,initial text=, initial left, font=\large]
        \node[state] (v0) {$v_0$};
        \node[state] (v1) [above right of=v0] {$v_1$};
        \node[state] (v2) [ right of=v1]  {$v_2$};
        \node[state] (v3) [ right of=v2]  {$v_3$};
        \node[state] (v4) [below right of=v0]  {$v_4$};
        \node[state,accepting] (t) [right of=v4, node distance={2*\nodedist}] {$t$};
        \node[] (vec1) [above left of=v1, node distance=\vecdist] {$\chargevec{4}{0}$};
        \node[] (vec4) [below left of=v4, node distance=\vecdist] {$\chargevec{0}{3}$};

        \path (v0) edge (v1);
        \path (v1) edge (v2);
        \path (v2) edge (v3);
        \path (v3) edge (t);

        \path (v0) edge (v4);
        \path (v4) edge (t);
        \path (v1) edge [loop above] (v1);
        \path (v2) edge [loop above] (v2);
        \path (v3) edge [loop above] (v3);
        \path (v4) edge [loop above] (v4);

        \end{tikzpicture}
        }
    \caption{Adding rewards changes strategy}
    \label{fig:second}
    \end{minipage}
\end{figure}

We note that a related variant of continuous bidding games with rewards, termed \term{bidding games with charging} (\txt{\term{CBGc}})~\cite{AvniGHM24}, has been studied in the literature.
In this setting, after the outcome of each round of bidding, a \emph{normalization} step is performed so that the sum of the players' budgets is rescaled to~$1$.
Avni et al.~\cite{AvniGHM24} show that a central property of classical bidding games extends to this setting: for every vertex, there exists a threshold ratio that characterizes the minimal fraction of the total budget required for a player to win from that vertex.
However, the situation is more intricate than in standard bidding games.
While in reachability continuous bidding games (reachability \txt{CBG}) the thresholds correspond to unique fixed points of linear systems of equations, in \txt{CBGc} these fixed points are no longer unique.
Moreover, Avni et al.~\cite{AvniGHM24} show that removing the normalization step—as we do in our model—leads to fundamentally different dynamics.
In particular, the threshold of one player becomes a non-linear function of the opponent's budget (see~\cite[Ex.~4]{AvniGHM24}).

From a modeling perspective, the normalization step in \txt{CBGc} enforces a constant total budget, effectively coupling the players' resources.
In contrast, in many applications, resources are accumulated and consumed independently by each agent.
Our model captures such scenarios more directly: budgets evolve through local updates (bids and rewards) without global rescaling.
This leads to a more faithful representation of systems in which resources such as energy, credits, or time are replenished and expended over time.

From a technical perspective, removing normalization also gives rise to qualitatively different phenomena.
As illustrated by our examples, the presence of rewards may change not only the outcome of the game but also the structure of optimal strategies, and leads to new algorithmic challenges in computing thresholds.


For example, play of a \txt{DPBGr} is more akin to a game played on an infinite graph. One may think of a discrete bidding game (with rewards or without) as taking place on an infinite graph of configurations (where a configuration holds the current vertex and current player budgets). Due to the decaying nature of budgets in \txt{DPBG}, the configuration visited in the infinite graph are necessarily finite. On the other hand, in \txt{DPBGr} budgets may grow and as such infinitely many configurations may be visited. 
In \txt{DPBG} the number of steps required to reach a target can be assumed to be polynomial in the size of the graph.
On the other hand in \txt{DPBGr}, budgets may grow in a variety of ways, and this polynomial bound no longer holds.

In this work, we show that for \txt{DPBGr} with a reachability objective, we may still compute the threshold by examining a finite segment of the infinite configuration graph. 
This is achieved by considering only plays in which both players play reasonably optimally---specifically, by avoiding \emph{disadvantaged loops} (see def. \cref{disadvantaged-loops}). Under this constraint we show that if a play is sufficiently long, optimal bids from that point on approximate optimal strategies for \term{continuous poorman bidding games} (\txt{CPBG}) which have been studied in \cite{Richman1999}, where thresholds were shown to exist and optimal strategies were described. 
Though, we show that this finite segment is exponential, there is reason to believe that this is not a tight bound.   
Finally, we examine a subset of \txt{DPBGr} in which \Roxy has access to a unique strategy, which we term a $B$-growth strategy (for some $B\in \bb{N}$). In such games, $th_v(B_1) \le B$ for all $B_1 \in \bb{N}$.

\section{Preliminaries}
\label{sec:prelim}

\subparagraph{Sequences notations} 
For a finite or infinite sequence $\textbf{c}=c_1,c_2,\ldots$ we use $\textbf{c}[i]$ for the index-$i$ element $c_i$, and $\textbf{c}[..i]$ for the prefix up to index $i$, inclusive.  We use $|\textbf{c}|$ for the length of $\textbf{c}$. Accordingly, if $\textbf{c}$ is finite, then $\textbf{c}[|\textbf{c}|]$ is the last element of $\textbf{c}$. 

Given a graph $(V,E)$ we use $\mathterm{\rightsquigarrow}$ to denote the Reachability relation. That is, for $u,v\in V$ we have $u \rightsquigarrow v$ iff there exists a path from $u$ to $v$. We extend this to sets, so $\mathterm{\pathbetween{S}{T}}$ if there exists a path from some $s\in S$ to some $t\in T$.  
Finally we use $\mathterm{\connTo{S}{T}}$ for the set $\{s\in S ~|~\pathbetween{s}{T}\}$.

\subparagraph{Bidding Games} 
A (two player) \term{bidding game} is a tuple $\cl{G}=(V,E,\Win,\bidmech)$, where $(V,E)$ is a directed graph, $\Win \subseteq V^\omega$ is a winning objective, and $\bidmech$ is a bidding mechanism which consists of 3 parameter - 
\begin{enumerate}
    \item \textbf{Bid granularity} -- is bidding discrete or continuous, meaning over $\bb{N}$ or $\bb{R}_{\ge0}$ respectively.
    \item \textbf{Update mechanism} -- What happens after bids take place. In Richman the winner pays the loser, in poorman the bid is paid to the bank and in all-pay everyone pays the bank. 
    \item \textbf{Tie breaking mechanism} -- How ties are resolved. For example, all ties are broken in favor of player $0$.
\end{enumerate}
\begin{remark}
    While in discrete bidding games tie breaking mechanisms are necessary, in continuous bidding games as bids can be arbitrarily small one may assume ties do not occur. 
\end{remark}
\begin{remark}
    Bidding games are not necessarily determined and can be impacted by the choice of tie breaking mechanism as was shown in \cite{determinacy} for discrete bidding games. For example, bidding games with alternating tie-breaking, in which ties are broken in turn, are not determined.
    Throughout this paper we implicitly assume only tie-breaking mechanisms for which \txt{DPBG} are determined.
\end{remark} 

A \term{configuration} of the game is a triple $(v,B_0,B_1)$, where $v\in V$ is the current vertex, and $B_0$ and $B_1$ are the budgets of Players~$0$ and~$1$, respectively.

An \term{initialized game} consists of a game $\cl{G}$ together with an initial configuration $(v,B_0,B_1)$. For a vertex $u\in V$ we use $\mathterm{succ(u)}$ for the set of successors of $u$, namely $\{ v \mid (u,v)\in E \}$.

\subparagraph{Bids, Bidding Rounds}
We use $\iota\in\{0,1\}$ for the player identity, and $\overline{\iota}$ for his opponent.
Consider a configuration $(v,B_0,B_1)$. 
A \term{bid} of Player~$\iota$ is a pair $(v_\iota,b_\iota)$, where $v_\iota$ is a successor of $v$ that Player~$\iota$ wishes to move the token to, and $b_\iota\leq B_\iota$ is the amount the player is willing to bid for this move.

A \term{bidding round} consists of the bids $(v_0,b_0)$ and $(v_1,b_1)$ of the players. 
If $b_\iota  > b_{\overline{\iota}}$, player $\iota$ wins the round. In the case of a tie $\bidmech$ determines who wins the round. 
Further, $\bidmech$ determines how a configuration $(v,B_0,B_1)$ is updated to a new configuration $(v',B'_0,B'_1)$. 
If Player~$\iota$ wins the bidding, then $v'=v_\iota$. The update budgets $B'_0$ and $B'_1$ depend on the bidding mechanism.
We henceforth consider only \term{poorman bidding}, in which the budgets are updated according to
$B'_\iota = B_\iota-b_\iota$ and 
$B'_{\overline{\iota}} = B_{\overline{\iota}}$, where $\iota$ (resp. $\overline{\iota}$) denotes the player who wins (resp. loses) the round.
That is, the winner pays its bid to the bank, while the loser's budget remains unchanged. We use $\mathterm{(v,B_0,B_1)\xrightarrow[u_1,b_1]{u_0,b_0}(v',B'_0,B'_1)}$ to denote that configuration $(v,B_0,B_1)$ was updated to $(v',B'_0,B'_1)$ due to bidding round consisting on bids $(u_0,b_0)$ and $(u_1,b_1)$ for player $0$ and $1$ respectively, and refer to this as a \term{valid configuration update}.

\subparagraph{Plays, Histories} 

An \term{extended play} in a bidding game  is a tuple $\explayseq=(\ppath,\pB{0},\pB{1},\pseq{u}{0}, \pseq{b}{0}, \pseq{u}{1},\pseq{b}{1})$ of infinite sequences $\ppath,\pseq{u}{0},\pseq{u}{1}\in V^\omega$, $\pB{0},\pB{1},\pseq{b}{0},\pseq{b}{1}\in\bb{N}^\omega$ where for each $i\in\bb{N}$
$$ 
(\ppath[i],\pB{0}[i],\pB{1}[i])\xrightarrow[\pseq{u}{1}[i+1\text{]},\ \pseq{b}{1}[i+1\text{]}]{\pseq{u}{0}[i+1],\ \pseq{b}{0}[i+1]}(\ppath[1],\pB{0}[i+1],\pB{1}[i+1]) 
$$
is a valid configuration update.

Any prefix of $\explayseq$ is referred to as a \emph{history} of $\explayseq$. We denote by $\mathterm{\History(\cl{G})}$ the set of histories of extended plays in $\cl{G}$.

\subparagraph{Strategies, Winning Strategies} 
In the case of bidding games, a strategy for player $\iota$ specifies not only the desired successor but also the offered bid.
    Consider a bidding game $\cl{G}=(V,E,\Win, \bidmech)$. A \term{strategy} $\sigma: \History(\cl{G}) \to V\times \bb{R}$ for player $\iota$ is a function where if $\sigma(\explayseq[..i])=(u,b)$ then $u\in succ(\ppath[i])$ and $0\le b \le B_\iota$.

An extended play $\explayseq$ is \term{consistent} with strategies $\sigma_0$ and $\sigma_1$ if $\sigma_\iota(\explayseq[..n])=(u_\iota[n+1],b_\iota[n+1])$ for $\iota\in\{0,1\}$. 
From a given configuration $(v,B_0,B_1)$ there is a unique play consistent with $\sigma_0$ and $\sigma_1$, we denote it $\mathterm{\langle \sigma_0,\sigma_1, v, B_0,B_1\rangle}$. When the initial configuration is clear we simply write $\mathterm{\langle \sigma_0,\sigma_1\rangle}$. 

An extended play $\explayseq$ is consistent with strategy $\sigma_\iota$ for player $\iota$ if there exists a strategy $\sigma_{\overline{\iota}}$ for player $\overline{\iota}$ such that $\explayseq$ is consistent with $\sigma_\iota$ and $\sigma_{\overline{\iota}}$. We denote by $\mathterm{\langle \sigma_0, \cdot\rangle}$ (resp. $\mathterm{\langle \cdot, \sigma_1\rangle}$) the sets of all extended plays consistent with $\sigma_0$ for player $0$ (resp. $\sigma_1$ for player $1$).

A strategy $\sigma$ for player $0$ (resp. player $1$) from initial configuration $(v,B_0,B_1)$ is a \term{winning strategy} if any $\explayseq \in \langle \sigma,\cdot\rangle$ (resp. $\explayseq \in \langle \cdot, \sigma\rangle$)
satisfies $\ppath \in \Win$ (resp. $\ppath \notin \Win$). 

\begin{remark}
    While we consider bidding games (that are concurrent), if a given strategy is winning for player $\iota$, since there exists no winning counter strategy, even if she reveals her action first she would still win. Thus, when considering winning strategies, we often consider a turn-based version of the game in which player $\iota$ reveals her bid first. For a game $\cl{G}$, we denote the respective trun-based game by $\cl{G}_\iota$. For example, see \cite[Theorem 3.2]{determinacy}.
\end{remark}

A central question of interest in bidding games is the \emph{threshold budget}, that is, what is the necessary and sufficient budget required for a player to win from a given vertex. More explicitly,

\begin{problem}[Threshold Budget]\label{prob:threshold-budget}
    Let $\cl{G}$ be a bidding game (with a rule set that results in a determined game), $v\in V$ and $B_1\ge 0$. The \term{threshold budget} asks what the minimal budget $B_0$ such that player $0$ has a winning strategy from $(v,B'_0,B_1)$ for $B'_0 \ge B_0$ is. We use $\mathterm{th_v(B_1)}$ to refer to the required threshold budget $B_0$.
\end{problem}
When considering continuous bidding games, any pair of configurations $(v,B_0,B_1)$ and $(v,\alpha B_0,\alpha B_1)$ are equivalent in the sense that scaling up every bid by $\alpha$ results in identical outcomes. As such, thresholds scale linearly with budgets (meaning $\alpha th_v(B_1)=th_v(B_1)$). Hence, it is sensible to consider thresholds as share of total budgets required to win (see \cite{IDPBG}).

\begin{problem}[Threshold Ratio]\label{prob:threshold-ratio}
    Let $\cl{G}$ be a continuous bidding game (with a rule set that results in a determined game), $v\in V$ and $B_1\ge 0$. The \term{threshold ratio} at $v$ is the minimum fraction of the total budget that the reachability player requires to win from $v$.
    We denote this threshold ratio  by $\mathterm{Th(v)}$.
\end{problem}

In \cite{Richman1999}, threshold ratios for both continuous Richman and poorman reachability games were solved. These results were extended in \cite{idbg, IDPBG} to Parity and Mean-Payoff objectives for continuous Richman games and poorman games respectively. These games are all shown to be in $\NP\cap \coNP$, though adding some restrictions on these games have been shown to have lower complexity. For instance, continuous Richman parity games in which the out-going degree of each vertex is at most 2 have been shown to be in $\PP$.

\section{Bidding Games with rewards}

\begin{definition}[Bidding game with rewards]

A \term{bidding game with rewards}, abbreviated \term{\txt{BGr}}, is a tuple $\cl{G}=(V,E,r_0,r_1,\Win,\bidmech)$ where $(V,E,\Win,\bidmech)$ are as in a bidding game and  $r_\iota:V\to \bb{N}$ are reward functions for each player $\iota\in\{0,1\}$.
\end{definition}
Configurations, bidding rounds, and strategies in bidding games with rewards are defined much the same as those in bidding games.
The difference is in the update mechanism, where after bidding takes place, if the token is moved to vertex $v$ and budgets are adjusted accordingly, the budget of each player $\iota$ is additionally increased by $r_\iota(v)$.

While this presentation of a bidding game is finite, it is often useful to imagine bidding games as games on an infinite graph. In this game vertices are configurations (i.e. include the position on the graph and  the budgets of both players), and edges connect configurations according to the behavior of the bidding process. 

Introducing rewards in this way is intuitive and offers an important benefit for some bidding mechanisms. 
While in Richman bidding, in the absence of rewards, the sum of budgets never changes, using other mechanisms (such as poorman or all-pay), the sum of the budgets decreases. 
This decay makes some objectives less meaningful. By introducing rewards, the sum of the budgets may increase.
For example, in a discrete poorman bidding game, without rewards, with a B\"uchi objective, arriving at a target vertex repeatedly may erode the budget of one player making the game entirely dependent on tie-breaking mechanism. 
Once rewards are introduced, even if budgets are depleted, they may increase again. This eliminates the above phenomenon.


As shown in~\cite{AlfaroH00}, determinacy is not guaranteed in concurrent games.
Determinacy of bidding games (which form a special class of concurrent games) has been studied in~\cite{determinacy}.
It is shown there that bidding games (without rewards) with reachability, B\"uchi, and parity objectives are determined, provided that the tie-breaking mechanism is unaware of when ties occur~\cite[Theorem~4.5]{determinacy}.
If this assumption is violated—for example, if ties are resolved by a transducer that receives as input information about tie occurrences—then the game may fail to be determined; see~\cite[Example~1.2]{determinacy}. The proof technique provided in \cite{determinacy} for \txt{DPBG} can be used to prove M\"uller objectives \txt{DPBGr} are determined (see \cref{dpbgr-determined} for a sketch of the proof).

Another desired property of bidding games with rewards is that it generalizes bidding games nicely. Meaning, a bidding game is equivalent to a bidding game with rewards where $r_\iota(v)=0$ for every $v\in V$.
As in other bidding games, a main research problem is the \term{threshold problem}, namely, given a \txt{BGr} $\mathcal{G}$, initial vertex $v$, an initial budget for player $1$, $B_1$, and a proposed threshold $\nu$, is $\nu\ge th_v(B_1)$.

Henceforth, we focus on \term{discrete poorman bidding games with rewards}, abbreviated \term{\txt{DPBGr}}, with a reachability objective.

We will refer to the player with reachability objective as \Roxy and her adversary (whose objective is safety) as \Sam. Further, to avoid confusion, we will henceforth index where appropriate using $\iota\in \{r,s\}$.

For convenience, we will assume that \Roxy's target is a single vertex $t$, $t$ has a single outgoing edge to itself and that $r_\iota(t)=-1$. None of these factors impact thresholds or strategies beyond labeling as in reachability games once a target vertex is reached the infinite suffix of the game no longer matters.

We now turn to solving the threshold problem for reachability \txt{DPBGr}.
We achieve this gradually. We first introduce \term{disadvantaged loops} (see \cref{sec:disadvantaged-loops}) and examine their impact on the outcome of a game. Next, we consider games starting from a configuration where both players start with very large budgets. Finally, we proceed to show that thresholds may be computed in \EXP (see \cref{sec:approx}).

\section{Disadvantaged loops}\label{sec:disadvantaged-loops}
A key determinant of ones ability to solve the threshold problem relies on the ability to determine whether a player has a winning strategy after finitely many steps. Players using bad strategies can extend play indefinitely without benefiting them. Different assumptions can be made on player strategies to help mitigate such issues. 
In this section we introduce the notion of \term{disadvantaged loops}. Any finite play that forms such a loop benefits only one player. We will further show that such plays must be the result of suboptimal play for one of the players. This will be crucial to show a bound on duration of games in which \Roxy has a winning strategy.

\begin{definition}[Disadvantaged loop]
\label{disadvantaged-loops}
     Let $\cl{G}$ be a \txt{DPBGr} and 
     $\explayseq$ be some valid extended play. Let $i<j\in \bb{N}$ such that $\ppath[i]=\ppath[j]$. We say the loop $\ppath[i] \ppath[i+1]\dots \ppath[j]$  is \term{\Sam disadvantaged} if either $\pBo[i] = \pBo[j]$ and $\pBz[i] > \pBz[j]$, or $\pBo[i] < \pBo[j]$ and $\pBz[i] \ge \pBz[j]$. We similarly define \term{\Roxy disadvantaged loops} except we consider also loops in which $\pBz[i] = \pBz[j]$ and $\pBo[i] = \pBo[j]$ as \Roxy disadvantaged.
\end{definition}

The case where none of the players' budget changes is considered disadvantaged for \Roxy, since her objective is reachability, and staying put serves the safety objective.

\begin{definition}[Non-disadvantaged loop]
\label{nondisadvantaged-loops}
     Let $\cl{G}$ be a \txt{DPBGr} and 
     $\explayseq$ be some valid extended play. Let $i<j\in \bb{N}$ such that $\ppath[i]=\ppath[j]$. We say $\ppath[i] \ppath[i+1]\dots \ppath[j]$ is an \term{upward loop} (resp. \term{downward loop}) if $\pBz[i] > \pBz[j]$ and $\pBo[i] > \pBo[j]$ (resp. $\pBz[i] < \pBz[j]$ and $\pBo[i] < \pBo[j]$).
\end{definition}

\begin{example}
Consider the \txt{DPBGr} presented in \cref{fig:diadvantagedfig}. The loop formed by the play
\[
(v_0,10,3)\xrightarrow[v_1,0]{v_1,0}(v_1,13,3)\xrightarrow[v_0,2]{v_2,1}(v_0,13,3)
\]

is a \Sam disadvantaged loop since \Roxy's budget has increased and \Sam's hasn't. 
On the other hand, the loop 
\[
(v_0,10,3)\xrightarrow[v_1,0]{v_1,0}(v_1,13,3)\xrightarrow[v_2,3]{v_0,4}(v_0,9,5)
\]
 is \Roxy disadvantaged since \Roxy's budget has decreased and \Sam's increased. 
Finally, the loop 
\[
(v_0,10,3)\xrightarrow[v_1,0]{v_1,0}(v_1,13,3)\xrightarrow[v_2,0]{v_2,0}(v_0,14,4)\xrightarrow[v_0,2]{t,1}(v_0,14,4)
\]
is an upward loop. 
    \begin{figure}
    \centering
    \scalebox{0.75}{
    \begin{tikzpicture}[->,>=stealth',shorten >=1pt,auto,node distance=\nodedist,semithick,initial text=, initial left, font=\large]
        
        \node[state] (v0) [] {$v_0$};
        \node[] (vec0) at ([xshift = -0.75 cm, yshift = -0.75 cm] v0) {$\chargevec{0}{2}$};
        \node[state] (v1) [right of=v0] {$v_1$};
        \node[] (vec1) at ([xshift = -0.75 cm, yshift = -0.75 cm] v1) {$\chargevec{3}{0}$};
        \node[state] (v2) [right of=v1] {$v_2$};
        \node[] (vec2) at ([xshift = -0.75 cm, yshift = -0.75 cm] v2) {$\chargevec{1}{1}$};
        \node[state,accepting]  (t) [right of=v2]  {$t$};
        
        \path (v0) edge (v1);
        \path (v1) edge  (v2);
        \path (v2) edge (t);
        \path (v2) edge [bend right = 45] (v0);
        \path (v1) edge [bend right] (v0);
        
        \end{tikzpicture}
        }
    \caption{}
    \label{fig:diadvantagedfig}
         
    \end{figure}
    
\end{example}

\begin{claim}[suboptimal strategy]
\label{claim:disadvantaged-strategy}
    Let $\cl{G}$ be a reachability \txt{DPBGr} and $(v,B_r,B_s)$ be some initial configuration. Suppose \Roxy and \Sam play according to strategy $\sigma_r$ and $\sigma_s$ respectively. Let $\explayseq= \langle\sigma_r, \sigma_s\rangle$. If $\explayseq$ contains a \Roxy (resp. \Sam) disadvantaged loop, and there exists a strategy $\sigma'_r$ (resp. $\sigma'_s)$ such that $\explayseq' = \langle\sigma'_r, \sigma_s\rangle$ contains no disadvantaged loops, then $\sigma_r$ (resp. $\sigma_s$) is a suboptimal strategy. 
\end{claim}
\begin{proof}
    This is a straightforward combinatorial result. 
    Suppose $i<j$ are the indices of a \Roxy disadvantaged loop. If $\pBz[i] < \pBz[j]$, \Sam has no fewer options available to him at configuration $(\ppath[j],\pBz[j], \pBo[j])$ than at $(\ppath[i],\pBz[i], \pBo[i])$. On the other hand \Roxy does have fewer options. 
    Similarly, if $\pBz[i] = \pBz[j]$ and $\pBo[i] > \pBo[j]$, \Roxy has the same number of options while \Sam has more options. 
    Finally, if $\pBz[i] = \pBz[j]$ and $\pBo[i] = \pBo[j]$, the number of options remain the same for both players but repeating this loop benefits \Sam as if the loop doesn't contain $t$, he wins and if it doesn't, he already lost before the loop was formed. 
    Similar arguments hold for \Sam disadvantaged loops. 

    Note that we require the existence of alternative strategies as there exist games in which one player can force the other to form a disadvantaged loop. 
\end{proof}

The above claim is useful as we may assume that if an alternative strategy to one that forms a disadvantaged loop exists, it is taken. Further, if no alternative exists and \Roxy can force \Sam disadvantaged loops, if a path to $t$ exists \Roxy clearly has a winning strategy. If no path to $t$ exists, she has already lost. Similarly, if \Sam can force \Roxy disadvantaged loops (these cannot visit $t$ since we assume $\ro(t)=-1$ and $t$ only has a path to itself) clearly \Roxy cannot win.

As such, when analyzing a play $\explayseq$ we may make a determination (assuming players play reasonably) on the outcome when considering a loop where $i<j$ and $\ppath[i]=\ppath[j]$ (for $\ppath[i]$ such that $\pathbetween{\ppath[j]}{t}$ otherwise clearly \Roxy loses).  

\begin{center}
    \begin{tabular}{c | c | c | c }
   & $\pBo[i] > \pBo[j]$ & $\pBo[i] = \pBo[j]$ & $\pBo[i] < \pBo[j]$ \\ \hline
   $\pBz[i] > \pBz[j]$ & unknown & \Sam wins & \Sam wins \\
   $\pBz[i] = \pBz[j]$ & \Roxy wins & \Sam wins & \Sam wins \\
   $\pBz[i] < \pBz[j]$ & \Roxy wins & \Roxy wins & unknown 
\end{tabular}
\end{center}

\begin{observation}
\label{obs:longplay=largebudgets}
  As we assume that players avoid disadvantaged loops when possible, we may deduce that a sufficiently long play must result in player budgets growing very large. More explicitly, suppose, a \txt{DPBGr} takes place on a graph with $12$ vertices (excluding $t$) from a starting configuration $(v,10,10)$. If play doesn't visit $t$ and doesn't include disadvantaged loops (for either player), every pair of configurations visited $(v_1,x_0,x_1),~(v_2,y_0,y_1)$ must differ in at least one index. Thus a pigeonhole argument shows that such a play lasting more than $1200$ bidding rounds must visit some configuration in which the budget of one of the players is at least $11$. Extending this argument to much longer plays can guarantee at least one players budget grows very large.  
\end{observation}

In the next section we will show a solution for \txt{DPBGr}. This will be achieved by approximating continuous poorman bidding games (henceforth, \txt{CPBG}). As such, we would like to extend the notion of disadvantaged loops further in settings that use continuous bidding. 

\begin{notation}[ratio-shift of $p$]
    Let $\cl{G}$ be a reachability \txt{CPBG} and $\explayseq$ be some valid play on it. 
    Let $i<j$ such that $\ppath[i] = \ppath[j]$ and $\pBo[i], \pBo[j]$ are both positive. We denote by $\mathterm{\Delta_{\explayseq[i..j]}}=\frac{\pBz[j]}{\pBo[j]}-\frac{\pBz[i]}{\pBo[i]}$ the \term{ratio-shift} of $\explayseq$.
\end{notation}

\begin{claim}
\label{CPBG-disadvantaged-loops}
    Let $\cl{G}$ be a reachability \txt{CPBG} and $\explayseq$ be some valid play on it. 
    Let $i<j$ such that $\ppath[i] = \ppath[j]$ and $\pBo[i], \pBo[j]$ are both positive. If $\Delta_{\explayseq[i..j]}\le 0$ ($\Delta_{\explayseq[i..j]} > 0$) then $\explayseq[i..j]$ is a \Roxy (\Sam) disadvantaged loop. 
\end{claim}
\begin{proof}
    As CPBG allow for arbitrarily small bids, a pair of configurations $(v,B_r,B_s)$ and $(v,B'_r,B'_s)$ such that $\frac{B_r}{B_s} = \frac{B'_r}{B'_s}$ are equivalent. This is because bids may be scaled appropriately. As such, we may consider $(\ppath[i], \pBz[i], \pBo[i])$ and $(\ppath[j], \pBz[j], \pBo[j])$ as $(\ppath[j], \frac{\pBz[j]}{\pBo[j]},1)$ and $(\ppath[j], \frac{\pBz[j]}{\pBo[j]},1)$ respectively. Using the usual definition for disadvantaged loops we achieve the desired result.
\end{proof}
\begin{remark}
\label{only-disadvantaged-loops-CPBG}
    The above claim implies that every loop in a \txt{CPBG} is disadvantaged for either \Roxy or \Sam (except when both players have an initial budget of $0$). Further, we may deduce that in a determined \txt{CPBG}, if player $\iota$ has a winning strategy she has one for which every play consistent with it is $\Bar{\iota}$ disadvantaged.
\end{remark}

\section{Approximation of \txt{DPBGr} via \txt{CPBG}}
\label{sec:approx}

In section 3.5 of \cite{DiscreteRichman2010}, Develin and Payne explore the intutition that discrete Richman bidding games with a large total budget behave a lot like continuous Richman bidding games. We show that for reachability objectives the same holds for \txt{DPBGr}. 
Recall that \cref{obs:longplay=largebudgets} states that when players play reasonably well long plays result in large player budgets. 
Combining these results we observe that after sufficiently long play, optimal bids start becoming more predictable. 

We use this to show that when both players play reasonably well (by avoiding disadvantaged loops and in the latter part of the game approximations of disadvantaged loops, see \cref{same-neighbor-threshold=smallbids}), there exists a hard upper bound on the budgets both players can hold. This surprising result can be used to show that if \Roxy has a winning strategy she can guarantee a win within a number of bidding rounds exponential in the number of vertices in the graph.

\begin{notation}
    Let $\cl{G}=(V,E,r_0,r_1,\Win,\bidmech)$ be a reachability \txt{DPBGr}. We denote by $\mathterm{\cl{G}_C} = (V,E,\Win,\bidmech')$ the continuous poorman bidding game played on the same graph as $\cl{G}$ with the same objective. (Here $\bidmech'$ is a bidding mechanism in which bids are continuous compared to $\bidmech$ which is discrete.)  
\end{notation}

\begin{notation}
    Let $\cl{G}$ be a \txt{CPBG}. We denote by $\mathterm{v^-}$ and $\mathterm{v^+}$ neighbors of $v$ with the lowest threshold ratio and highest threshold ratio respectively.
\end{notation}
\begin{remark}
    $v^+$ and $v^-$ need not be unique, there may exist multiple neighbors of $v$ with minimal or maximal threshold ratios. We may assume that one such neighbor is chosen arbitrarily. 
    As was shown in \cite{Richman1999}, when optimal strategies are used in a \txt{CPBG} \Roxy always chooses to bid to move to a neighbor with minimal threshold ratio and \Sam to one with maximal threshold ratio. 
\end{remark}

In order to show that we may compute the threshold we consider what happens when play persists for many rounds of bidding. Here ``many" is no more than exponential in the number of vertices in the graph of the game, and the term ``grow large" used below will be made more precise in  \cref{sufficiently-large-budgets-reach}.
\begin{enumerate}
    \item Budgets grow large and \Roxy's share of total budget is greater than $Th_{\cl{G}_C}(v)+\varepsilon$. In this case \Roxy wins. See \cref{sufficiently-large-budgets-reach}
    \item Budgets grow large and \Roxy's share of total budget is less than $Th_{\cl{G}_C}(v)-\varepsilon$. In this case \Roxy loses. See \cref{sufficiently-large-budgets-safety}
    \item Budgets grow large and \Roxy's share of total budget is between $Th_{\cl{G}_C}(v)-\varepsilon$ and $Th_{\cl{G}_C}(v)+\varepsilon$. In this case we may compute the winner. See \cref{clm:roxy-win-duration-bound}  
    \item Budgets do not grow large. In this case disadvantaged loops have formed and can be ignored. See \cref{obs:longplay=largebudgets}.
\end{enumerate}

Before examining each of these cases we make an observation regarding the behavior in a continuous bidding setting.

\begin{observation}
[\txt{CPBG} winners stays above threshold]
\label{clm:above-threshold-CPBG}
Let $\cl{G}$ be a \txt{CPBG}. Suppose \Roxy (resp. \Sam) has a winning strategy $\sigma$ from configuration $(v,B_r,B_s)$ then any play $\Pi$ consistent with $\sigma$ has $\frac{\pBz[i]}{\pBz[i]+\pBo[i]}>Th(\ppath[i])$ (resp. $\frac{\pBz[i]}{\pBz[i]+\pBo[i]}<Th(\ppath[i])$) for all $i\in \bb{N}$.
\end{observation}
\begin{proof}
    We show the proof for \Roxy, the proof for \Sam is largely the same. 
    Suppose \Sam has a strategy $\tau$ such that for $\langle \sigma, \tau\rangle$, there exists some $i\in \bb{N}$ such that $\frac{\pBz[i]}{\pBz[i]+\pBo[i]}<Th(\ppath[i])$ (note that we assume the equality cannot occur since bids are continuous). 
    Then \Sam may use a strategy $\tau'$ that plays as $\tau$ until configuration $(\ppath[i],\pBz[i],\pBo[i])$ and pivots to a winning strategy from it (which exists due to the budget ratio being below the threshold for $\ppath[i]$). Thus $\langle \sigma, \tau'\rangle $ is a winning play for \Sam which contradicts the assumption that $\sigma$ is winning.
\end{proof}

We may now proceed to show that in \txt{DPBGr} that start with large initial budgets, optimal bids are approximations of their corresponding bids in \txt{CPBG}. 

We begin by examining the cases in which \Roxy's share of the total budget is far from the threshold. 

\begin{notation}
    For convenience, we denote by $\mathterm{R_\iota} = \sum_{v\in V\setminus \{t\}}r_\iota(v)$, i.e. the total sum of rewards for player $\iota$ (excluding the target).
\end{notation}

\begin{notation}
    We use $\mathterm{\cl{B}_n(B,R)}$ for  the expression $3B+3(3R+1)^2(3R+2)^{n-2}$.
    Given a \txt{DPBGr}, we think of $n$ as the number of vertices in the graph. We think of 
    $B$ and $R$ as \emph{initial budget} and the \emph{sum of all rewards of a specific player}, respectively.
\end{notation}

\begin{restatable}[Large budget approximation of \txt{CPBG} winning strategy]{claim}{sufficientlylargebudgetsreach}\label{sufficiently-large-budgets-reach}
Let $\cl{G}$ be a reachability \txt{DPBGr}. Fix some $1>\varepsilon>0$. For sufficiently large budgets $\Br$, at most exponential in number of vertices, if $\frac{\Br}{\Br+\Bs}\ge Th_{\cl{G}_C}(v)+\varepsilon$ then \Roxy has a winning strategy from $(v,B_r,B_s)$ in the original game $\cl{G}$ that she wins within $exp(n)$ bidding rounds when $\varepsilon$ is given in a number of bits at most linear in $n$. 
\end{restatable}
\begin{proof}[Proof Sketch]
    In \txt{CPBG} when playing a winning strategy bids are given by arbitrarily more than the optimal bid $\Delta(v)$.  
    In $\cl{G}_C$ we consider \Roxy's budgets as split into a "real" budget $Th(v)(\Br+\Bs)$, used to make the optimal part of a bid and a "slush fund" $\varepsilon(\Br+\Bs)$, used to overpay slightly to produce a true winning strategy.
    To win, \Roxy plays in a way that grows the slush fund to a point where she has a large enough budget to win consecutive bids until $t$ is visited. Every time \Sam wins a bid \Roxy's share of the total budget increases by a constant $d$ that is dependent on $\varepsilon$. Once her budget is sufficiently large to win at most $n$ consecutive bids and visit $t$ she does so.
    This entire strategy takes a number of rounds exponential in $n$ when $\varepsilon$ is given in a number of bits at most linear in $n$.

    We show that if $\frac{\Br}{\Br+\Bs}\ge Th_{\cl{G}_C}(v)+\varepsilon$, for $\Br$ such that $\varepsilon > \cl{B}_n(\Bs, \Ro)$ (as $\varepsilon$ and $\Br$ are fixed, $\Bs$ is a function of $\Br$), we may approximate such a strategy in the continuous game $\cl{G}_C$ sufficiently well and still produce a winning strategy. 
\end{proof}

\begin{notation}
\label{test}
    We use $\mathterm{\mathfrak{Br}_{\varepsilon}}$ to denote the budget required by \Roxy to approximate a winning strategy in the setting of \cref{sufficiently-large-budgets-reach}.  
\end{notation}

\begin{corollary}
\label{sufficiently-large-budgets-safety}
    Let $\cl{G}$ be a reachability \txt{DPBGr}. Fix some $\varepsilon>0$. For sufficiently large $\Br$, if $\frac{\Br}{\Br+\Bs}\le Th_{\cl{G}_C}(v)-\varepsilon$ then \Sam has a winning strategy from $(v,\Br,\Bs)$. 
\end{corollary}
\begin{proof}
    In \cite[Theorem 7]{IDPBG}, an equivalence between reachability and \emph{double-reachability}, where both players have a target, is shown. It is established by setting the target for \Sam to be the set of vertices that do not have a path to $t$.
    If no such vertex exists, then $Th(v)=0$ for every $v\in V$ (see  in \cite[Proposition 7]{idbg}). Thus, this case is not relevant to the claim.
    Hence, we we may assume that if the set $\pathbetween{V\not}{t}$ is non-empty then a winning play eventually visits one of these vertices. 
    This results in a proof similar to that of \cref{sufficiently-large-budgets-reach}. 
\end{proof}

\begin{remark}
    Since optimal bids are presented as share of current budget, we deduce may similarly denote by $\mathfrak{Bs}_\varepsilon$ the budget sufficient for \Sam to approximate optimal play much like is the case in the setting of \cref{sufficiently-large-budgets-reach} (with roles reversed). 
\end{remark}

The cases where \Roxy's share of the budgets is far from the threshold are fairly straightforward since one player wins soon after budgets grow large enough. 

\begin{notation}
    We denote by $\mathterm{\mathfrak{B}}=\max \{\mathfrak{Br}_\varepsilon, \Rz\mathfrak{Bs}_\varepsilon\}$. It follows from \cref{sufficiently-large-budgets-reach} and \cref{sufficiently-large-budgets-safety} that if \Roxy has $\mathfrak{B}$ budget at $v$ and her (resp. his) share of the budget is at least (resp. at most) $Th(v)+\varepsilon$ (resp. $Th(v)-\varepsilon)$, then she (resp. he) has a winning strategy.
    Note here that $\Rz\mathfrak{Bs}_\varepsilon$ guarantees \Sam's budget is more than $\mathfrak{Bs}_\varepsilon$ when no disadvantaged loops have formed. This is because if we assume his budget grows as slow as possible compared to \Roxy, he must gain $1$ coin for every $\Rz$ coins she gains.
\end{notation}

When addressing the case where her share is close to the threshold, we must contend with the challenge that we can no longer guarantee the games conclusion. That being said, we can make observations regarding what optimal bids in a single round of bidding are.

\begin{corollary}[Optimal bids near threshold for non-flat vertex]
\label{diff-neighbor-threshold=bigbids}
    Let $\cl{G}$ be a reachability \txt{DPBGr}. For some sufficiently small $\varepsilon>0$ and sufficiently large $\Br$, if $\big|\frac{\Br}{\Br+\Bs}-\varepsilon\big|\le Th_{\cl{G}_C}(v)$ and $Th_{\cl{G}_C}(v^+) \neq Th_{\cl{G}_C}(v)$ (resp.  $Th_{\cl{G}_C}(v^-) \neq Th_{\cl{G}_C}(v)$) then the optimal bid for \Roxy (resp. \Sam) is greater or equal to $\Br \cdot\varepsilon$.
\end{corollary}
\begin{proof}
    As the optimal bid in $\cl{G}_C$ is $\Delta(v)\Br$
    if \Roxy (resp. \Sam) makes a bid $b < \Br \cdot \varepsilon$ and \Sam (resp. \Roxy) bids $b+1 \le \Br \cdot \varepsilon$, we arrive at a configuration $(v^+,\Br',\Bs')$ (resp. $(v^-,\Br',\Bs')$) in which, $\frac{\Br}{\Br+\Bs} < Th_{\cl{G}_C}-\varepsilon$ (resp. $\frac{\Br}{\Br+\Bs} > Th_{\cl{G}_C} + \varepsilon$) using \cref{sufficiently-large-budgets-safety} (resp. \cref{sufficiently-large-budgets-reach}) we deduce that \Sam (resp. \Roxy) has a winning strategy. 
\end{proof}

\begin{definition}
    Let $\cl{G}$ be a reachability \CPBG. We call a vertex $v\in V$ such that $Th(v)=Th(v^-)=Th(v^+)$ 
    \mathterm{flat}.
\end{definition}

\begin{corollary}[Optimal bids near threshold for flat vertex]
\label{same-neighbor-threshold=smallbids}
    Let $\cl{G}$ be a reachability \txt{DPBGr}. Fix some $\varepsilon>0$ and sufficiently large, if $v$ is flat and $\big|\frac{\Br}{\Br+\Bs}-\varepsilon\big|\le Th_{\cl{G}_C}(v)$ then the optimal bid for \Roxy and \Sam is less or equal to $\Br \cdot \varepsilon$.
\end{corollary}
\begin{proof}
    Suppose \Roxy makes a bid of more than $\Br \cdot \varepsilon$.  
    As $Th_{\cl{G}_C}(v^+) = Th_{\cl{G}_C}(v)$ we must deduce that $Th_{\cl{G}_C}(v)=Th_{\cl{G}_C}(u)$ for any $u\in succ(v)$. The optimal bid for either player in $\cl{G}_C$ is $0$. If \Sam bids $0$ we move to a new configuration with the same threshold but \Roxy's share of the total budget has reduced by $\varepsilon$ thus \Sam has a winning strategy. The case for \Sam is similar. 
\end{proof}

When the conditions of \cref{diff-neighbor-threshold=bigbids} are met, players are forced to make large bids and therefore despite rewards, total budget must decrease. On the other hand, if the conditions of \cref{same-neighbor-threshold=smallbids} are met, small bids must be made. How small is ambiguous and may lead to upward loops forming. As such, if enough of the graph consists of flat vertices, this may lead to budgets continuing to increase. 
We show that in this case, loops formed approximate disadvantaged loops in \txt{CPBG} (as seen in \cref{CPBG-disadvantaged-loops}). Therefore, if players play reasonably well, \Roxy or \Sam must choose to play another strategy.

\begin{claim}
\label{bigbudget-disadvantaged-loops}
    Let $\cl{G}$ be a reachability \txt{DPBGr}. Fix some $\varepsilon>0$ and sufficiently large $\Br$. Let $\sigma_r$ and $\sigma_s$ be strategies for \Roxy and \Sam resp., and the play $\explayseq = \langle \sigma_r,\sigma_s\rangle$. Suppose an upward loop is formed in $\explayseq$ where $\ppath[i],\ppath[i+1],\dots , \ppath[j-1]$ are all flat. If $\frac{\pBz[j]-\pBz[i]}{\pBz[j] + \pBo[j] - (\pBz[i] + \pBz[i])} > Th_{\cl{G}_C}(v)$, then $\sigma_s$ is suboptimal. 
\end{claim}
\begin{proof}
    Observe that for $i \le k \le j-1$, the vertex $\ppath[k+1]$ is a successor of $\ppath[k]$ so $Th_{\cl{G}_C}(\ppath[k])= Th_{\cl{G}_C}(\ppath[k+1])$. Thus, by \cref{same-neighbor-threshold=smallbids}, reasonable bids made on these vertices are less than $(\Br + \Bs)\varepsilon$.
    Assuming players play sufficiently well we may assume only reasonable bids are made. 

    Hence, since $\Br$ and $\Bs$ are large, every reasonable bid is already available to both players before performing any upward loop. Thus the added bidding options are meaningless to both players.
    On the other hand, if this loop is allowed to repeat, the ratio of player budgets will eventually surpass $Th_{\cl{G}_C}(\ppath[i])$ resulting in a win for \Roxy. 
\end{proof}

\begin{corollary}
\label{bigbudget-disadvantaged-loops-safety}
    Let $\cl{G}$ be a reachability \txt{DPBGr}. Fix some $\varepsilon>0$ and sufficiently large $\Br$. Let $\sigma_r$ and $\sigma_s$ be strategies for \Roxy and \Sam resp., and the play $\explayseq = \langle \sigma_r,\sigma_s\rangle$. Suppose an upward loop is formed in $\explayseq$ where $\ppath[i],\ppath[i+1],\dots , \ppath[j-1]$ are all flat. If $\frac{\pBz[j]-\pBz[i]}{\pBz[j] + \pBo[j] - (\pBz[i] + \pBz[i])} \le Th_{\cl{G}_C}(v)$, then $\sigma_r$ is suboptimal.
\end{corollary}
\begin{proof}
    The proof is essentially the same as that of \cref{bigbudget-disadvantaged-loops}.
\end{proof}

\begin{remark}
\label{budget-ceiling-rmk}
    The above claims show that if both players play reasonably well, one of them must force a downward loop to form. Otherwise, \Roxy's share of total budgets will drift far enough from the threshold to enable \cref{sufficiently-large-budgets-reach} or \cref{sufficiently-large-budgets-safety} to be used.zzzz
\end{remark}

\begin{claim}
\label{budget-ceiling}
    Let $\cl{G}$ be a reachability \txt{DPBGr}. For $\varepsilon>0$, if \Roxy and \Sam play reasonably well, there exists a ceiling on the budgets of both players.
\end{claim}
\begin{proof}
    Consider a configuration $(v,\mathfrak{B}_\varepsilon, \Ro\mathfrak{B}_\varepsilon)$. 
    \cref{diff-neighbor-threshold=bigbids}, shows that if a vertex is not flat, in order to approximate a winning strategy, players must make bids of more than $\Br \cdot \varepsilon > \Rz + \Ro + 1$ guaranteeing a downward loop is formed.
    In the case of \cref{same-neighbor-threshold=smallbids} either a downward loop is formed, which is sufficient for the claim, or an upward loop is formed. In the latter case, the budget of player $\iota$ may increase by at most $\Ri$.
    \cref{bigbudget-disadvantaged-loops} shows that one player must either force both players to spend more than they receive in rewards or eventually lose the game.
\end{proof}

\begin{notation}
    We denote the ceiling on player budgets by $\mathterm{\cl{C}}$
\end{notation}

\begin{claim}
    \label{clm:roxy-win-duration-bound}
    Let $\cl{G}$ be a reachability \txt{DPBGr} and $(v,\Br',\Bs')$ be some initial configuration. If \Roxy has a winning strategy then the duration of the game is at most $2n \cl{C}^2$.
\end{claim}
\begin{proof}
    \cref{budget-ceiling} implies that when both players play reasonably well budgets do not surpass $\cl{C}$. 
    As such, any play persisting for more than $n\cl{C}$ rounds of bidding must contain disadvantaged loops. If the loops are \Roxy disadvantaged she will definitely lose the game. Otherwise, the loops are \Sam disadvantaged. If \Sam has been playing reasonably well, every loop to form from here on out must be disadvantaged. Assume \Roxy benefits the least amount possible (meaning every $n\cl{C}$ rounds of bidding a single disadvantaged loop forms and her budget grows by $1$ while \Sam's remains the same, in this case her share of the budget increases by the slowest rate possible). Then after at most $n\cl{C}$ rounds of bidding her budget will have grown by at $\cl{C}$, therefore her share of the total budget will have grown by more than $\varepsilon$ and she can approximate a winning strategy, which takes less than $n\cl{C}$ more rounds of bidding, totaling less than $2n\cl{C}^2$.
\end{proof}

\begin{theorem}
\label{evaluation-onebudget}
    Let $\cl{G}$ be a reachability \txt{DPBGr} and $(v,\Br',\Bs')$. If budgets are given in binary, verifying if $\Br' \ge th_v(\Bs')$ is in \EXP. 
\end{theorem}
\begin{proof}
\cref{clm:roxy-win-duration-bound} shows that the duration of a game if \Roxy has a winning strategy is at most $2n\cl{C}^2$ rounds of bidding. We may check if a winning strategy exists by considering all possible ending configuration for a game whose duration is $2n\cl{C}^2$ rounds of bidding, if the configuration is in vertex $t$, clearly \Roxy has won, otherwise she lost. Computing the outcome of the previous round relies entirely on the results from the set of configurations already computed. We may propagate this result to eventually decide if the initial configuration has a winning strategy. 
As the number of possible ending configuration after $2n\cl{C}^2$ is at most exponential we achieve the desired result.
\end{proof}

\begin{corollary}
    Let $\cl{G}$ be a reachability \txt{DPBGr}. If bids and budgets are given in binary, computing $th_v(\Bs)$ is in \EXP.
\end{corollary}
\begin{proof}
    As the threshold for a reachability \txt{DPBGr} has $0\le th_v(\Bs) \le n(\Bs+\Ro+1)$, we simply use \cref{evaluation-onebudget} polynomially many times.
\end{proof}

\section{Growth-Strategy}
\label{sec:growth-strategy}

In this section we examine games in which \Roxy has a distinct advantage of \Sam. In such games, \Roxy has two advantages, one fairly obvious and the other is, at first glance, not an advantage at all. 
The first advantage is that every vertex must have a path to the target. This is a clear advantage as \Roxy can then never be punished too severely for playing poorly. 
The second advantage is that there exists an upper bound on how fast \Sam's budget may grow in relation to \Roxy. Intuitively, this is not really an advantage as if \Sam's budget grows many times faster he should benefit greatly. The reason this is an advantage is a direct result of \cref{sec:approx}. When budgets grow sufficiently large, a \txt{DPBGr} approximates a \txt{CPBG}. In CPBG, if every vertex is connected to the target \Roxy has a winning strategy with threshold $0$ from any vertex. Since we consider only reasonable strategies for \Roxy (that do not contain \Roxy disadvantaged loops), her budget is guaranteed to be far greater than the threshold for a sufficiently long play.

Throughout this section we assume for convenience that ties are broken in favor of \Sam (which is a tie breaking mechanism for with \txt{DPBGr} are determined). 
Note that for any other tie-breaking mechanism \Sam must make larger bids, as such, these results hold for bidding games with any determined tie-breaking mechanism.

We present below $B$-growth strategies and show that in \txt{DPBGr} in which \Roxy has such a strategy, $th_v(B_s)\le B$ for all $B_s \in \bb{N}$.

\begin{notation}
    Given a sequence $\pbold{u}=\pbold{u}[1]\ \pbold{u}[2]\ \pbold{u}[3] \ldots \in V^\omega$  we denote by $\mathterm{I_v^p}$ the set of indices such that \(i\in I^\pbold{u} _v \iff \pbold{u}[i]=v\). 
    For $i \in I^\pbold{u}_v$, we write $\mathterm{succ(i)}$ to indicated the minimal index $j\in I^\pbold{u}_v$ that is greater than $i$.
\end{notation}

\begin{definition}[$B$-growth strategy]\label{def:b-growth}
    Consider a \txt{DPBGr}. For $v\in V$,  a strategy $\sigma$ for \Roxy is a \term{$B$-growth strategy for $u$} for $B\in \bb{N}$ if any extended play $\explayseq$ starting at configuration $(v,B,x)$ for any $x\in\bb{N}$, consistent with $\sigma$, for some $m>1$ the play has the following properties : 
    \begin{enumerate}[itemsep=2ex] 
        \item \term{Loop Non-Decreasing Property: }\label{prop:1} $\forall i \in I^\ppathseq_v : ~\pBz[succ(i)]\ge \pBz[i]$.  
        \item  \term{Loop Ratio-Increase Bound:}\label{prop:2} $\forall i\in I^\ppathseq_v. ~\pBz[succ(i)]>\pBz[i] \implies   ~\frac{|\pBo[succ(i)]-\pBo[i]|}{\pBz[succ(i)]-\pBz[i]} < m$. 
        \item \term{Loop Zero-Stagnation Property: }\label{prop:3} $\mbox{$\forall i\in I^\ppathseq_v .~\pBz[succ(i)]=\pBz[i] \implies   ~\pBo[succ(i)]-\pBo[i]<0$}$.  
        \item \term{Reachability Property:}\label{prop:4} $\forall i\in \bb{N}. ~v_i \in \connTo{V}{t}$.  
    \end{enumerate}
\end{definition}

\begin{remark}
    Recall we assume that $t$ has only self edges and has $\ro(t) = -1$. As such, \Sam cannot simply elect to lose the game causing the definition to never apply.
\end{remark}

\begin{example}
\label{growth-strat-ex}
    In the game shown in \cref{fig:firstgame} Roxy has a $2$-growth strategy $\sigma$ from $v_0$. At vertex $v_0$ she bids $2$ to move to $v_1$, at $v_1$ she bids $3$ to move to $t$.

    Now suppose play begins from configuration $(v_0,2,x)$ for some $x\in \bb{N}$. Suppose \Roxy plays according to $\sigma$. If \Sam outbids \Roxy at $v_0$ to remain, he must bid at least $2$, when configurations are updated the new configuration will be $(v_0,2,x-1)$, this works with property 3. 

    Suppose \Roxy wins, the updated configuration is $(v_1,3,x+6)$. If \Roxy wins her bid at $v_1$ we move to $t$ and every loop formed will have property 3. Otherwise, \Sam bids at least $3$ to win and we move to a configuration $(v_0, 3, x+3)$. The loop produced $(v_0,2,x) \to (v_1,3,x+6) \to (v_0, 3, x+3)$ has ratio growth of $m=3$ (as in property 2). Any other loop produced must have this ratio or worse (if \Sam bids more than $3$ at $v_1$), we deduce that $m=3$ is as in the definition.
    
    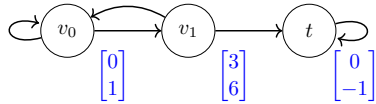
\begin{figure}[h]
    \label{growth-strat-game}
    \centering
   \begin{center}
    \scalebox{0.8}{
        \begin{tikzpicture}
        [ 
        square/.style={regular polygon,regular polygon sides=4, inner sep=0.25 cm, draw}
        ] 

            \node[state] (v0) at (0,0) {$v_0$};
            \node[] (vec0) at ([xshift = +0.75 cm, yshift = -0.75 cm] v0) {$\chargevec{0}{1}$};
            \node[state] (v1) at (2,0) {$v_1$};
            \node[] (vec1) at ([xshift = +0.75 cm, yshift = -0.75 cm] v1) {$\chargevec{3}{6}$};
            \node[state] (t) at (4,0) {$t$};
            \node[] (vec2) at ([xshift = +0.75 cm, yshift = -0.75 cm] t) {$\chargevec{0}{-1}$};

            \path[->, line width = 0.7pt] (v0) edge [loop left] node [] {} (v0); 
            
            \path[->, line width = 0.7pt] (v0) edge [] node [] {} (v1);
            
            \path[->, line width = 0.7pt] (v1) edge [bend right = 30] node [] {} (v0);
            
            \path[->, line width = 0.7pt] (v1) edge [] node [] {} (t);

            \path[->, line width = 0.7pt] (t) edge [loop right] node [] {} (t);           
        \end{tikzpicture}
        }
    \end{center}
    \caption{Example of a \txt{DPBGr} with a $2$-growth strategy}
    \label{fig:firstgame}
\end{figure}

\end{example}

\begin{claim}
\label{games without growth strat}
    Let $\cl{G}$ be a Reachability \txt{DPBGr}. For any $u\in V$, if there exists a vertex $v$ such that $\pathbetween{u}{v}$ and $\pathbetween{v\not}{t}$ then there exists no $B$-growth strategy for $u$.
\end{claim}
\begin{proof}
    For any initial configuration $(u,B,x)$, $B\in \bb{N}$ if $x \ge |V|\cdot (B+1)$ \Sam may proceed to win consecutive bids to visit $v$. This means the reachability property (\cref{def:b-growth}.\cref{prop:4}) of growth strategies cannot hold. 
\end{proof}

We proceed to show that in any Reachability \txt{DPBGr} $\cl{G}$ starting at some vertex $v_0$ for which \Roxy has a growth-strategy, \Roxy wins. 


\begin{restatable}[A growth strategy implies a winning strategy]{theorem}{Bgrowthwinningstrat}\label{Bgrowth-winning-strat}
    Let $\cl{G}$ be a reachability \txt{DPBGr}. 
    If \Roxy has a $B$-growth strategy $\sigma$ for $v_0$, then \Roxy has a winning strategy for any initial configuration of the form $(v_0,\Br,\Bs)$ where $\Br\ge B$.
\end{restatable}
\begin{proof}[Proof Sketch]
    Intuitively, the proof relies on the idea that given \Roxy has a growth strategy, she risks nothing in losing any single bid. We use this to construct a strategy that forces \Sam to spend much more than \Roxy does, forcing him to deplete his budget. To achieve this, \Roxy plays according to her growth strategy to achieve a budget of $\cl{B}_n(\Bs,\Ro)$. At which point she can play a strategy that forces \Sam to deplete his budget at a rate $2\Ro+1$ times faster than \Roxy. Once \Sam's budget is reduced sufficiently she wins consecutive bids to visit $t$.
\end{proof}

Since the existence of a growth-strategy implies every vertex is connected to the target, $\cl{G}_C$ has threshold $0$ for every vertex (\cite[Lemma 6]{Avni2021-sd}). Since budgets eventually grow very large \Roxy's budget is well above the threshold thus she eventually wins.   

\begin{corollary}
    Let $\cl{G}$ be a reachability \txt{DPBGr} such that for all $v\in V$, $\pathbetween{v}{t}$ and $\rz(v)>0$. Then $th_v(\cdot)=0$ for all $v$. 
\end{corollary}
\begin{proof}
    \Roxy's growth strategy is to bid $0$ always. In this case every time a loop is formed her budget will have grown by some quantity and \Sam's can grow by at most $\Ro$ thus the growth ratio is bound. 
\end{proof}

This result hints that in games in which \Roxy's rewards are dispersed more equally across all vertices, her threshold may be generally lower. 

We now show that checking if a $B$-growth strategy exists is substantially easier than solving a \txt{DPBGr}. 

\begin{restatable}[Large budget approximation of \txt{CPBG} winning strategy]{theorem}{BgrowthNP}\label{B-growth-NP}
    Let $\cl{G}$ be a reachability \txt{DPBGr}, deciding if \Roxy has a $B$-growth strategy from $v$ is in $\NP$.
\end{restatable}

\section{Discussion}
We examine for the first time discrete bidding games with rewards and have shown that the threshold problem for a reachability \txt{DPBGr} is in \EXP. 
We have shown that if a play is allowed to persist for sufficiently long, when players play reasonably well, optimal bids begin to resemble continuous bidding games. While membership in \EXP is somewhat disappointing compared to bidding games without rewards,  there is reason to believe that this is not a lower bound on complexity. This is because, one would expect that as budgets grow larger optimal bids should grow more similar to optimal bids in \txt{CPBG}. In our work we only show a point after which bids are close approximations. This may be an interesting direction for future work. 

We introduce a novel technique to be able to enforce that players play reasonably well (using disadvantaged loops). This method is applicable to other bidding games and may be useful when considering other generalizations of bidding games as well as other objectives.

Finally we show that in some arenas the reachability player holds a substantial advantage.

\bibliographystyle{plainurl}
\bibliography{bib}

\appendix
\section{Omitted Proofs}

\begin{claim}
\label{dpbgr-determined}
    Let $\cl{G}$ be a M\"uller \txt{DPBGr} that uses a tie-unaware transducer for tie breaking. Then $\cl{G}$ is determined.
\end{claim}
\begin{proof}[Proof Sketch]
    Aghajohari et al. \cite{determinacy} defines a fragment of concurrent games, $R$-concurrent games, for a transducer $R$ that, given a configuration, determines the available actions. Vertices on $R$-concurrent games are partitioned into configuration vertices denoted $C$ (which in the setting of bidding games can be thought of as $(v,B_0,B_1,s)$ where $v$ is the current vertex $B_\iota$ are current player budgets and $s$ is the state of the tie breaking mechanism) and intermediate vertices, denoted $I$ (which are used for book keeping). See 3.2 in \cite{determinacy} for more precise details. 

    It is straightforward to show that a \txt{DPBGr} is $R$-concurrent (for a transducer $R$ with infinitely many states). 

    Theorem 3.7 of \cite{determinacy} shows that for a locally determined $R$ (see \cite[def. 3.4]{determinacy}), any $R$-concurrent game with a M\"uller objective is determined. 

    \cite[subsection 3.3]{determinacy} defined the bidding matrix $M_c$ for a  configuration $c=(v,B_0,B_1,s)$ whose rows and columns are allowed actions for player $0$ and $1$ respectively. Examining $M_c$ we find that 
    $(M_c)_{ij}=0$ iff player $0$ has a winning strategy in the game in which she reveals her action first. Since turn based M\"uller games are determined this is well defined. 

    Using this they show that discrete bidding games are $R$-concurrent for a locally determined $R$ and are therefore determined. 

    For \txt{DPBGr} with a M\"uller objective, we may define the bidding matrix similarly. Note that the matrix is still well defined since the turn-based version of a \txt{DPBGr} is a turn-based game with a Borel objective and by \cite{BorelDeterminacy} is determined.

    Finally they define a tie-unaware transducer and show that discrete bidding games whose ties are resolved by a tie-unaware transducer are $R$-concurrent for a locally determined $R$ and are therefore determined. The same proof can be used for \txt{DPBGr}.  
\end{proof}

\begin{lemma}
\label{explicit-winning-CPBG}
    Let $\cl{G}$ be a Reachability \txt{CPBG} and initial configuration $(v,\Br,\Bs)$ be such that $\frac{\Br}{\Br+\Bs}=Th(v)+\varepsilon$ for $\varepsilon>0$. Then \Roxy has a winning strategy such that she wins within $poly(\frac{n}{\varepsilon})$ bidding rounds.
\end{lemma}
\begin{proof}
    The proof is based on the ideas presented in \cref{Bgrowth-winning-strat}. Roughly speaking, the strategy relies on \Roxy's ability to guarantee budgets become large enough (a growth-strategy, see \cref{def:b-growth}) to employ a winning strategy. This strategy resembles the strategy presented here in a setting where the continuous threshold is $0$ everywhere.
    
    As budgets may be scaled linearly, we may assume for convenience that our initial configuration is $(v,Th(v)+\varepsilon,1-Th(v)-\varepsilon)$.
    When playing winning strategy on a \txt{CPBG} , bids are made arbitrarily larger than the optimal bid, which is given by
    $$
    \mathterm{\Delta(v)} = \bigg(\frac{1}{Th_{\cl{G}_C}(v)}-\frac{1}{Th_{\cl{G}_C}(v^+)} \bigg)
    $$
    We consider $Th(v)$ and \Roxy's ``real'' budget (used to bid the optimal part of the bid) and $\varepsilon$ as a slush fund (used to pay the small increase).
    We construct an explicit winning strategy for \Roxy such that every time \Sam wins a round of bidding, \Roxy's share of the total budget increases compared to the last time he won until she wins. Set $\mathterm{\cl{B}}=\cl{B}_n(\Bs,1)$. 

    \Roxy's strategy is as follows 
    \begin{enumerate}[label=(\roman*)]
        \item Let $\pbold{u}[0],\dots \pbold{u}[k]$ where $\pbold{u}=v$ and $\pbold{u}[k]=t$ be a shortest path such that $\pbold{u}[i+1]\in \pbold{u}[i]^-$ for $0\le i< k$. At $\pbold{u}[0]$ \Roxy bids $\Delta(\pbold{u}[0]) + b_0$ where $b_0 = \frac{\varepsilon}{\cl{B}}$ to move to  $\pbold{u}[1]$. At $\pbold{u}[1]$ she bids $\Delta(\pbold{u}[1]) + b_1$ where $b_1 = \frac{3\varepsilon}{\cl{B}}$ to move to $\pbold{u}[2]$ and for $i>1$ at $\pbold{u}[i]$ she bids $\Delta(\pbold{u}[i])+b_i$ where $b_i = \frac{3\cdot 4^{i-1}\varepsilon }{\cl{B}}$ to move to $\pbold{u}[i+1]$.
        \item If \Roxy wins every bid, she wins the game. Otherwise, if \Roxy loses at $\pbold{u}[i]$ and \Sam chooses successor $u$, if \Sam's share of the total budget is more than $\frac{1}{2n+1}$, set $v\leftarrow u$, set the configurations so that budgets are $(u, Th(u)+\varepsilon',1-Th(u)-\varepsilon')$ repeat \emph{(i)}. Otherwise move to \emph{(iii)}.
        \item \Roxy bids $\frac{2}{(2n+1)}$ every bid (and wins since the bid is larger than \Sam's entire budget) until she visits $t$.  
    \end{enumerate}

    \begin{observation}
    \label{obscont}
        At \emph{(ii)} if a bid is won by \Sam at $\pbold{u}[j]$ and the token is moved to vertex $u$, \Roxy's budget is $Th(u) +\varepsilon - \sum_{i=0}^{j-1} b_i$. On the other hand the budget of \Sam has decreased by $b_j$. Note that $b_j = 3\sum_{i=0}^{k-1}b_i$. Therefore we deduce that the budget of \Sam is $Th(u)-\varepsilon-3\sum_{i=0}^{j-1} b_i$. Meaning every time we repeat step \textit{(i)} the budget of \Sam has decreased at least $3$ times more than \Roxy and \Roxy's slushfund has increased by at least $2\sum_{i=0}^{j-1} b_i$. Every time \emph{(i)} is repeated \Roxy's slush fund grows faster. Take $d = 2b_0$, after \Sam wins $\frac{2n}{(2n+1)d}$ rounds of bidding \Roxy's slush fund represents $\frac{2n}{2n+1}$ of the total budget so the strategy is well defined. 
    \end{observation}

    Notice that the duration of the game is at most $\frac{2n^2}{(2n+1)d}+n$ since loops can only form if \Sam wins a round of bidding (since \Roxy wants a simple path to $t$ at \emph{(i)}) and so $\frac{2n}{(2n+1)d}$ loops form \Sam must have won that many times and a loop must form every $n$ rounds of bidding unless \Roxy wins.

\end{proof}

\begin{lemma}
\label{DPBG-approx}
    Let $\cl{G}$ be a reachability \txt{DPBG}. Fix some $1>\varepsilon>0$. For sufficiently large budgets $\Br$, if $\frac{\Br}{\Br+\Bs}\ge Th_{\cl{G}_C}(v)+\varepsilon$ then \Roxy has a winning strategy from $(v,B_r,B_s)$ in the original game $\cl{G}$ such that she wins within $poly(\frac{n}{\varepsilon})$ bidding rounds.
\end{lemma}
\begin{proof}
    Using \cref{explicit-winning-CPBG} we would like to show that one may approximate a winning strategy for \Roxy on a \txt{DPBG}. 
    Similar to the proof of \cref{explicit-winning-CPBG}, we consider $\flBrack{Th(v)(\Br+\Bs)}$ as \Roxy's ``real'' budget and $\Br - \flBrack{Th(v)(\Br+\Bs)} = \clBrack{\varepsilon(\Br+\Bs)}$ as her slush fund (unlike \txt{CPBG} bids are in $\bb{N}$). 
    Before proceeding to produce a winning strategy we observe the behavior of bids in this setting. 
    
    Let $\delta  >\max\{\frac{1}{\Br}, \frac{1}{\Bs}\}$. We consider a round of bidding from $(v,\Br,\Bs)$ where \Roxy bids $\flBrack{\Delta(v)\Br}+1$ to move to  instead of her true optimal bid $\Delta(v)\Br$.
    
    Suppose \Roxy loses her bid, then (assuming \Sam plays well) we visit configuration $(v^+,\Br , \Bs - \flBrack{\Delta(v)\Br}-1)$ and \Sam has slightly overbid.
    Computing the difference in \Roxy's share of the budget between the optimal bid and it's approximation   
    $$
    \frac{\Br }{\Br + \Bs - \flBrack{\Delta(v)\Br} -1} - \frac{\Br }{\Br + \Bs - \Delta(v)\Br }  
    $$
    $$
    \ge \frac{\Br }{\Br + \Bs - \Delta(v)\Br -\delta\Br} - \frac{\Br }{\Br + \Bs - \Delta(v)\Br }
    $$
    $$
    = \frac{\delta\Br^2 }{(\Br + \Bs - \Delta(v)\Br -\delta\Br)(\Br + \Bs - \Delta(v)\Br)} \ge \frac{\delta \Br^2}{\Br^2}=\delta
    $$
    We deduce that \Roxy's slush fund has grown by at least $\delta$.

    Similarly, if \Roxy wins, the difference in the approximated vs un-approximated outcomes is:
    $$
     \frac{\Br - \Delta(v)\Br }{\Br-\Delta(v)\Br + \Bs } - \frac{\Br - \flBrack{\Delta(v)\Br} - 1 }{\Br - \flBrack{\Delta(v)\Br} - 1 + \Bs} \le 
    $$
    $$
    \le \frac{\Br - \Delta(v)\Br }{\Br-\Delta(v)\Br + \Bs } - \frac{\Br - \Delta(v)\Br - 1 }{\Br -\flBrack{\Delta(v)\Br} - 1 + \Bs} \le \frac{1}{\Br - \Delta(v)\Br + \Bs}
    $$
    $$
    \frac{\delta \Bs}{\Br - \Delta(v)\Br + \Bs}\ge \frac{\delta\Bs}{\Bs} = \delta
    $$

    so \Roxy's slush fund decreases by no more than $\delta$. 

    Taking $\Br$ large enough so that $\varepsilon > \cl{B}$ we may use the the strategy presented in \cref{explicit-winning-CPBG} (bids must be modified to be $\flBrack{\Delta(\pbold{u}[i] + b_i)}$, since every bid is a slight underbid compared to the previous strategy \Roxy's slush fund (as share of total budget increases by at least $d$).
\end{proof}

\sufficientlylargebudgetsreach*

\begin{proof}
    Construction is similar to \cref{DPBG-approx}, except bids must be adjusted to allow for income. To do this we take $\cl{B} = \cl{B}_n(\Bs, \Ro)$ and take bids that are $\Ro+1$ times larger. 
\end{proof}

\Bgrowthwinningstrat*
\begin{proof}
     We construct a winning strategy for \Roxy as follows -- 
    \begin{enumerate} 
        \item \label{(i)} Play according to $\sigma$ until your budget is at least $\mathterm{\cl{B}}=\cl{B}_n(\Bs,\Ro)$ (recall that $\Ri = \sum_{v\in V\setminus \{t\}\ri(v)}$). Set the first vertex visited when this threshold is surpassed by $u$ (see \cref{obs1} below).  
        \item \label{(ii)} Let $\pbold{u}[0],\pbold{u}[1]\dots \pbold{u}[k]$ where $\pbold{u}[0]=u$ and $\pbold{u}[k]=t$ be a shortest path from $u$ to $t$. \Roxy proceeds to bid $b_0=(3\Ro+1)$ at $\pbold{u}[0]$ to visit $\pbold{u}[1]$, $b_1 = (3\Ro+1)^2 $ and for $i>1$, $b_i= b_{i-1}\cdot(3\Ro+2)$ at $\pbold{u}[i]$ to visit $\pbold{u}[i+1]$. 

        \item \label{(iii)} If \Roxy wins every bid, she wins the game. Otherwise, if \Roxy loses at $\pbold{u}[i]$ and the chosen successor is $v$, set $u \leftarrow v$ and repeat \textit{(ii)}.
    \end{enumerate}
    
    The proof that this strategy is winning makes use of the following observations.

    \begin{observation}
    \label{obs1}
        The loop Zero-Stagnation property of growth strategies (see \cref{def:b-growth}) implies that only finitely many consecutive loops in which the budget of \Roxy does not increase may occur. This is because  when her budgets remains the same upon revisiting some vertex $v$,  the budget of \Sam has decreased. Repeating this eventually results in the budget of \Sam being entirely depleted. Thus, we are assured the budget of \Roxy can grow arbitrarily large and the strategy does eventually proceed to (\ref{(ii)}).
    \end{observation}

    \begin{observation}
    \label{obs2}
        At (\ref{(iii)}) if a bid is won by \Sam at $\pbold{u}[j]$ and the token is moved to vertex $v$, the budget of \Roxy has decreased by at most $\sum_{i=0}^{j-1} b_i$ in the bids between $\pbold{u}[0]$ and $\pbold{u}_j$ (if \Roxy had received some rewards it decreases by less). On the other hand the budget of \Sam has decreased by at least $b_j - j\Ro$ (since for every $v$, $\ro(v)\le \Ro$ and $j$ vertices were visited from the first bid at $\pbold{u}[0]$ to the bid \Sam won). Note that $b_k = (3\Ro+1)\sum_{i=0}^{k-1}b_i$. Therefore we deduce that the budget of \Sam has decreased by at least $(3\Ro+1)(\sum_{i=0}^{k-1}b_i) - k\Ro \ge (2\Ro+1)(\sum_{i=0}^{k-1}b_i)$. Meaning every time we repeat step \textit{(ii)} the budget of \Sam has decreased at least $2\Ro+1$ times more than \Roxy. 
    \end{observation}

    We now show that the strategy is winning. Notice that when step \ref{(i)} is complete, due to property 2 and 3 of growth strategies (\cref{def:b-growth}), the budget of \Sam is at most $\Bs+m\cl{B}+2\Ro$ (The $2\Ro$ term is due to the path of play possibly not starting and ending with a loop). Since the strategy described forces \Sam to spend $2\Ro+1$ times more than \Roxy,
    once the budget of \Roxy has decreased by no more than $\frac{2\cl{B}}{3}$, the budget of \Sam has been decreased by at least 
    $$\frac{(4\Ro+2) \cl{B}}{3} =\Ro\cl{B} + \frac{(\Ro+2)\cl{B}}{3}
    $$
    $$
    =\Ro \cl{B}+ (\Ro+2) \Bs+(\Ro+2)(3\Ro+1)^2(3\Ro+1)^{n-2} \ge \Bs + \Ro \cl{B} + 2\Ro$$
    Which is more than \Sam's entire budget after completing step \ref{(i)}.
    On the other hand, the budget remaining for \Roxy is at least $(3\Ro)^2(3\Ro+1)^{n-2} = \sum_{i=0}^{n-1}b_i$ which is sufficient to repeat Step \ref{(ii)} one last time without losing any bids winning the game.
\end{proof}

\begin{lemma}
\label{lemma:all-vertices-have-growth-strat}
    Let $\cl{G}$ be a reachability \txt{DPBGr}, suppose \Roxy has a $B$-growth strategy from $v$. Then for any $u$, \Roxy has a $B_u$-growth strategy for some $B_u \le B+ \Rz$.
\end{lemma}
\begin{proof}
    Let $\sigma$ be a $B$-growth strategy from $v$ for \Roxy. 
    Suppose there exists $u\in succ(v)$ that has no $B_u$-growth strategy. Then for $B+\rz(u)$ there exists some budget $B_1$ for \Sam and strategy $\tau$ such that the play $\langle \sigma,\tau\rangle$ starting at $(u,B+\rz(u),B_1)$ contains loops that do not correspond to the properties of a growth-strategy. 
    
    Clearly when considering a play starting at $(v,B,B_1+B+1)$, \Sam may bid $B+1$ to move to $u$ and win. A similar argument can be made if $B_u > B+ \Rz$. Repeating this argument for neighbors of newly discovered vertices we find that every vertex have a $B'$-growth strategies. 
\end{proof}

\begin{lemma}
\label{lemma:growth-strat-is-positional}
    Let $\cl{G}$ be a reachability \txt{DPBGr}, suppose \Roxy has a $B$-growth strategy $\sigma$ from $v$. Then there exists a positional $B$-growth strategy $\sigma'$ from $v$.
\end{lemma}
\begin{proof}
    From \cref{lemma:all-vertices-have-growth-strat} we deduce that one must simply choose for every vertex the first bid of a growth-strategy to perform it. 
\end{proof}

\BgrowthNP*
\begin{proof}
    To show membership in NP, we use a witness, $B$ and a positional strategy (meaning a bid $b_u$ for every vertex $u$ and successor $s_u\in succ(u)$) as  shown exists in \cref{lemma:growth-strat-is-positional}. 
    To verify this result we simply add weights to the graph of the game for \Roxy and \Sam. 
    $$
    w_r(u,v)=
        \begin{cases}
    			\rz(v)-b_v & v=s_u\\
                \rz(v) & \text{otherwise}
    	\end{cases}
    $$

    $$
    w_s(u,v)=
        \begin{cases}
    			\ro(v) & v=s_u\\
                \ro(v)-(b_v+1) & \text{otherwise}
    	\end{cases}
    $$

    $W_r$ simulates the change in \Roxy's budget when she plays according to the strategy provided by the witness. $W_s$ simulates the change in \Sam's budget assuming \Sam chooses to outbid \Roxy with the smallest possible amount . (Recall that in this section we assume ties are broken in favor of \Sam, with minor adaptation a similar proof can be provided for other tie breaking mechanisms).
    Using a simple BFS we determine if there is a path for which either \Roxy's budget decreases or \Roxy's budget remains the same but \Sam's grows.
\end{proof}

\end{document}